\documentclass[letterpaper, 10 pt, conference]{ieeeconf}  

\usepackage[utf8]{inputenc}
\usepackage{fourier}
\usepackage{amsmath}
\usepackage[utf8]{inputenc}
\usepackage{amssymb}
\usepackage[utf8]{inputenc}
\usepackage{graphicx}
\usepackage{caption}

\usepackage{amsmath}

\usepackage{mathtools}
\usepackage[thinc]{esdiff}
\usepackage{xcolor}
\usepackage{epsfig}
\usepackage{epstopdf}
\usepackage{subcaption}
\usepackage[ruled,vlined]{algorithm2e}
\usepackage{array}

\usepackage[english]{babel}
\usepackage{amsthm}
\usepackage{amssymb}

\newtheorem{theorem}{Theorem}
\newtheorem{lemma}[theorem]{Lemma}
\newtheorem{definition}{Definition}

\IEEEoverridecommandlockouts                              

\IEEEoverridecommandlockouts                              

\title{\LARGE \bf
Computing Bounds on $L_{\infty}$-induced Norm for Linear Time Invariant Systems Using Homogeneous Lyapunov Functions
}

\author{Hassan Abdelraouf$^{1}$, Gidado-Yisa Immanuel$^{2}$ and Eric Feron$^{3}$
\thanks{*This work was supported by King Abdullah University of Science and Technology (KAUST)}
\thanks{$^{1}$Ph.D. candidate, Mechanical Engineering, KAUST, Thuwal, Saudi Arabia
{\tt\small hassan.abdelraouf@kaust.edu.sa}}%
\thanks{$^{2}$PhD candidate , School of Aerospace
Engineering, Georgia Institute of Technology, Atlanta, 30332, USA
        {\tt\small gidado.immanuel@gatech.edu}}%
\thanks{$^{3}$Professor, Department of Computer, Electrical and Mathematical Sciences and Engineering, KAUST, Thuwal, Saudi Arabia. 
{\tt\small eric.feron@kaust.edu.sa}}%
}

\begin{document}

\maketitle
\thispagestyle{empty}
\pagestyle{empty}

\begin{abstract}
Quadratic Lyapunov function has been widely used in the analysis of linear time invariant (LTI) systems ever since it has shown that the existence of such quadratic Lyapunov function certifies the stability of the LTI system. In this work, the problem of finding upper and lower bounds for the $L_{\infty}$-induced norm of the LTI system is considered. Quadratic Lyapunov functions are used to find the star norm, the best upper on the $L_{\infty}$-induced norm, by bounding the unit peak input reachable sets by inescapable ellipsoids. Instead, a more general class of homogeneous Lyapunov functions is used to get less conservative upper bounds on the $L_{\infty}$-induced norm and better conservative approximations for the reachable sets than those obtained using standard quadratic Lyapunov functions. The homogeneous Lyapunov function for the LTI system is considered to be a quadratic Lyapunov function for a higher-order system obtained by Lifting the LTI system via Kronecker product. Different examples are provided to show the significant improvements on the bounds obtained by using Homogeneous Lyapunov functions. 
\end{abstract}

\section{INTRODUCTION}

For single-input,single-output (SISO) continuous-time invariant systems, the ${L}_\infty$-induced norm, the signal peak-to-peak gain of a transfer function, is the $\ell_1$ norm that can be computed by $\int_{0}^{\infty}|h(t)|\text{dt}$, where $h(t)$ is the impulse response of the system. Shamma \cite{shamma1993nonlinear} shows that there does not exist a closed form expression for the  $\ell_1$ norm of the LTI system. Given the state space model of the system, an approximate value of the  $\ell_1$ norm can be obtained by only simulation. The problem is to find the final time that provides an accurate approximation for the integration value. This fact motivates the search for alternative methods to get reliable upper bounds for the $\ell_1$ norm.  Authors in \cite{abedor1996linear} introduce the \textit{star-norm} to provide a valid upper bound for the $\ell_1$ norm by the construction of  inescapable ellipsoids using quadratic lyapunov functions. But the \textit{star-norm} calculated in \cite{abedor1996linear} represents a very conservative upper bound on the $\ell_1$ norm for stiff systems, such as introduced in \cite{feron1994linear}. This motivated us to introduce new methods for computing more conservative approximations for the $\ell_1$ norm of linear time invariant systems. 

Several methods in the control theory literature are used in computing $\ell_1$ norm including the originating article \cite{vidyasagar1986optimal} in which the author computes the $\ell_1$ norm for both discrete and continuous time systems with some initial results on controller design for disturbance rejection.  In \cite{boyd1987comparison}, a comparison between the $L_{\infty}$ induced norm of a discrete system and its RMS gain (the maximum magnitude of the its frequency response) is made to see how this difference affects $\mathbf{H}_\infty$ optimal controller design. An improved upper and lower bounds for a discrete time system 's $\ell_1$ norm  are introduced in \cite{balakrishnan1992computing} . Based on these bounds, the worst case $L_\infty$ induced norm of the discrete-time systems with diagonal perturbations can be computed accurately. The $\ell_1$ optimal controller design for  discrete-time MIMO systems is  introduced in \cite{dahleh19871}. The controller is designed to make the system internally stable and optimally track a bounded input signal. The results in \cite{dahleh19871} are extended to $\ell_1$ optimal controller design for continuous time systems in \cite{dahleh19872}. The problem of $\ell_1$ optimal controller with full state feedback is considered in \cite{shamma1996optimization} where author showed that instead of using state feedback dynamic linear controller with arbitrary higher order, the $\ell_1$ optimal controller can be nonlinear memoryless static feedback.

In this note, we introduce different techniques in approximating the $\ell_1$ norm of LTI systems. The method used in \cite{abedor1996linear} depends on computing the \textit{star-norm} by using quadratic Lyapunov function to bound the reachable sets. This methods produces good bounds for some systems. But for stiff systems, as we will show later, these bounds become more conservative. Our approach is based on using higher order homogeneous Lyapunov function  in generating the inescapable ellipsoids to reduce the conservatism in the computed upper bound of the $\ell_1$ norm. This higher order Lyapunov function can be considered as a quadratic Lyapunov function for a higher order "lifted" system. The authors of \cite{abate2020lyapunov} show how the lifted systems are generated for linear time varying systems using a recursive algorithm based on the Kronecker product. In addition, they show that the higher order homogeneous Lyapunov function that certifies stability for linear varying  can be considered as a quadratic Lyapunov function for the lifted system to a higher degree. The homogeneous Lyapunov functions are also used to obtain better approximations for some performance metrics for linear time varying systems in \cite{abate2021pointwise}. The authors show that the bounds on system's peak norms  obtained using higher order homogeneous Lyapunov function are more accurate and less conservative than the bounds resulted using quadratic Lyapunov functions. In \cite{abate2021pointwise}, the effect of control inputs is considered in the lifting process. In this work, we extend the work introduced in \cite{abate2020lyapunov} and \cite{abate2021pointwise} and use homogeneous Lyapunov function to compute less conservative upper bounds for the linear time invariant system's $\ell_1$ norm. These bounds are less than the bounds introduced in \cite{abedor1996linear}. 
In this work, the reachable sets for LTI systems using unit peak input are  approximated using higher order homogeneous Lyapunov functions, then the star norm is computed based on the generated ellipsoid to get better approximation for the system's $\ell_1$ norm. Additionally, we introduce another technique to better approximate the integral value $\int_0^\infty |C e^{At} B| \text{dt}$ by calculating the integral to a specific time $T_0$ then approximate an upper bound for the remaining integral $\int_{T_0}^\infty |C e^{At} B| \text{dt}$ by computing the star-norm for a an equivalent new system. The upper bound for the remaining integral is also improved by lifting the new system. We show that it is more powerful that computing the star norm of the original system directly. To illustrate the results, we compute upper bounds on  $\ell_1$ norm  for different types of systems like: systems with high damping, systems with low damping and systems with stiff mass matrix, by using our proposed approaches that gives very accurate approximations compared with the methods introduced in the literature.

\section{Notation}
Denote the sets of non-negative and positive integers by $\mathbb{Z}_{+}$ and $\mathbb{Z}_{++}$ respectively. $\mathbb{R}_{+}$ and $\mathbb{R}_{++}$ denotes the set of non-negative and positive real numbers respectively. The set of positive definite $n \times n$ matrices is denoted by $\mathbb{S}_{++}^{n} \subset \mathbb{R}^{n \times n}$. For $P \in \mathbb{R}^{n \times n}$, $P \succ 0$ means that $P$ is a positive definite matrix and the function $V(x)=x^T P x$ is positive for all non-zero $x \in R^n $. The $n \times n$ identity matrix is denoted by $I_n$. 
\subsection{Kroncker product}
For matrices $A \in \mathbb{R}^{n \times m}$ and $B \in \mathbb{R}^{p \times q}$, the Kroncker product of $A$ and $B$ is denoted by $A \otimes B \in \mathbb{R}^{np \times mq}$ and is given by
\begin{equation}\label{kroncker product}
    A \otimes B:=\left[\begin{array}{ccc}
a_{11} B & \cdots & a_{1m} B \\
\vdots & \ddots & \vdots \\
a_{n1} B & \cdots & a_{nm} B
\end{array}\right]
\end{equation}
where $A$ matrix entries are represented via subscript. We will introduce some properties of the kroncker product defined in \cite{broxson2006kronecker} and \cite{van2000ubiquitous}. For $A \in \mathbb{R}^{n \times m }$, the $d \textsuperscript{th}$ kroncker power $A^{ \otimes d}$ for all $d \in \mathbb{Z}_{++}$ is defined recursively with a base $A^{\otimes 0} =1$ by: 
 \begin{equation}\label{kroncker sum}
 \begin{aligned}
     A^{\otimes 1}&= A \\
     A^{\otimes d}&= A \otimes A^{\otimes (d-1)}, \quad d=2,3, \dots
 \end{aligned}
 \end{equation}
 The following important properties of kroncker product are used in this work: 
 \begin{subequations}\label{kroncker property}
 \begin{align}
      A \otimes (B+C)= A \otimes B + A \otimes C \label{1}\\
      (A \otimes B)(C \otimes D) = (AC) \otimes (BD) \label{2}\\
      (AB)^{\otimes n}= A^{\otimes n} B ^{\otimes n} \label{3}
 \end{align}
 \end{subequations}
 where all matrices $A$,$B$,$C$ and $D$ are with proper dimensions that permit the formation of products $AC$ and $BD$.
 In \cite{broxson2006kronecker},corollary (7) states that property (\ref{2}) can be generalized as 
 \begin{multline}\label{4}
     A_{1} B_{1} \otimes A_{2} B_{2} \otimes \cdots \otimes A_{n} B_{n}\\=\left(A_{1} \otimes A_{2} \otimes \cdots \otimes A_{n}\right)\left(B_{1} \otimes B_{2} \otimes \cdots \otimes B_{n}\right),
 \end{multline}
 such that all matrices $(A_1,A_2,\dots A_n)$ and $(B_1,B_2,\dots B_n)$ are with proper dimensions that allows all multiplications in the right hand side.

The $d \textsuperscript{th}$ kroncker sum of a symmetric matrix $A \in \mathbb{R}^{n \times n}$ is defined by
 \begin{equation}\label{kroncker sum}
     A^{\oplus d}= \sum_{k=1}^d I_n^{\otimes(k-1)} \otimes A \otimes I_n^{(d-k)}.
 \end{equation}
 For example, 
 \begin{equation}\label{kroncker sum example}
 \begin{aligned}
     A^{\oplus 3}&= I_n^{\otimes 0} \otimes A \otimes I_n^{\otimes 2}  +I_n^{\otimes 1} \otimes A \otimes I_n^{\otimes 1} +I_n^{\otimes 2} \otimes A \otimes I_n^{\otimes 0} \\
     &= A\otimes I_n \otimes I_n + I_n \otimes A \otimes I_n + I_n \otimes I_n \otimes A.
 \end{aligned}
 \end{equation}
\subsection{Norms}
For a vector $v \in \mathbb{R}^n$, we define the ${L}_\infty$ norm as 
\begin{equation}\label{L infty norm for a vector}
 \|v\|_{\infty}=\sup _{i}\left|v_{i}\right|
\end{equation}
where $v= [v_1,v_2, \dots, v_n]^{T}$. The function space $\mathcal{L}_\infty$ is defined as the set of functions $f(t)$ for which the $L_{\infty}$ norm: 
\begin{equation}
    \|f\|_{\infty}=\sup _{t \in \mathbb{R}+}\|f(t)\|_{\infty}
\end{equation}
is bounded. The $L_{\infty}$-induced norm of an operator $\mathcal{G}$ which maps $u$ to $y=\mathcal{G}u$ is defined by 
\begin{equation}\label{L infity induced norm}
    \|\mathcal{G}\|_{\infty}=\sup _{\underset{u \neq 0}{u \in \mathcal{L}_{\infty}}} \frac{\|\mathcal{G} u\|_{\infty}}{\left\|u\right\|_{\infty}}
\end{equation}
Given that the operator $\mathcal{G}$ is considered as the following single-input single-output linear time invariant  system
\begin{equation} \label{LTI system}
    \mathcal{G}:\left\{\begin{array}{l}
\dot{x}=A x+B u \\
y=Cx
\end{array}\right.
\end{equation}
where $x \in \mathbb{R}^n$ and ($A \in \mathbb{R}^{n \times n},B \in \mathbb{R}^{n \times 1}$ and $C \in \mathbb{R}^{1 \times n}$ ) are the system's matrices, the $L_{\infty}$-induced norm in (\ref{L infity induced norm}) is the $\ell_1$ norm of the system in (\ref{LTI system}) which can be computed as 
\begin{equation} \label{l1 norm}
    \|\mathcal{G}\|_1 = \int_0^{\infty} |h(t)| \text{dt},
\end{equation}
where $h(t)$ is the impulse response of the system and can be computed as: 
\begin{equation}\label{impulse response}
    h(t)= C e^{At} B
\end{equation}

\section{Dynamical system lifting} 
In this section, we incorporate the results introduced in \cite{abate2020lyapunov} and \cite{abate2021pointwise} on lifting the system dynamics from the original space to a higher order space. Additionally, we consider the control input in the lifting procedure. 
\subsection{Lifting procedure}
First, given the LTI system in (\ref{LTI system}), the state vector $x$ is lifted to $x^{\otimes d}$ where $d$ is the lifting order. Then, the chain rule is applied to obtain the derivative of the lifted state vector $x^{\otimes d}$ with respect to time as 
\begin{equation}\label{lifting process 1}
\begin{aligned}
    \frac{d}{dt} x^{\otimes d} &= \sum_{k=1}^{d} x^{\otimes {(k-1)}} \otimes \dot{x} \otimes x^{\otimes (d-k)} \\
    &= \sum_{(k=1)}^{d} x^{\otimes {k-1}} \otimes (Ax+Bu) \otimes x^{\otimes (d-k)}. \\
\end{aligned}
\end{equation}
By using property (\ref{1}), equation (\ref{lifting process 1}) can be written as 
\begin{equation}\label{lifting process 2}
\begin{aligned}
    \frac{d}{dt} x^{\otimes d} &= \sum_{k=1}^{d} x^{\otimes {(k-1)}} \otimes Ax \otimes x^{\otimes (d-k)}\\
    &+\sum_{k=1}^{d} x^{\otimes {(k-1)}} \otimes Bu \otimes x^{\otimes (d-k)}.
\end{aligned}
\end{equation}
Since $x= I_n x$, property (\ref{3}) can be used. Therefore, (\ref{lifting process 2}) becomes
\begin{equation}\label{lifting process 3}
\begin{aligned}
    \frac{d}{dt} x^{\otimes d} &= \sum_{k=1}^{d} (I_n^{\otimes {(k-1)}} x^{\otimes {(k-1)}}) \otimes Ax \otimes (I_n^{\otimes (d-k)} x^{\otimes (d-k)})\\
    &+\sum_{k=1}^{d} (I_n^{\otimes {(k-1)}} x^{\otimes {(k-1)}}) \otimes Bu \otimes (I_n^{\otimes (d-k)} x^{\otimes (d-k)}),
\end{aligned}
\end{equation}
hence property (\ref{4}) can be applied  in (\ref{lifting process 3}) to be
\begin{equation}\label{lifting process 4}
\begin{aligned}
    \frac{d}{dt} x^{\otimes d} &= \sum_{k=1}^{d} 
    (I_n^{\otimes (k-1)} \otimes A \otimes I_n^{\otimes (d-k)}) x^{\otimes d}\\
    &+\sum_{k=1}^{d} (I_n^{\otimes (k-1)} \otimes B \otimes I_n^{\otimes (d-k)}) u x^{\otimes (d-1)}.
\end{aligned}
\end{equation}
By using (\ref{kroncker sum example}), (\ref{lifting process 4}) can be summarized as
\begin{equation}\label{lifted states}
    \frac{d}{dt} x^{\otimes d} = A^{\oplus d} x^{\otimes d} + B^{\oplus d} ux^{\otimes(d-1)}.
\end{equation}
The output equation in (\ref{LTI system}) is lifted as 
\begin{equation}\label{lifted output}
    y^{\otimes d } = C^{\otimes d } x^{\otimes d}.
\end{equation}
To summarize this section, the LTI system (\ref{LTI system}) can be lifted to a higher order space with a degree $d \in \mathbb{Z}_{++}$ as
\begin{equation}\label{ lifted system}
\begin{aligned}
    \dot{\zeta}&= A^{\oplus d} \zeta + B^{\oplus d} w\\
    \eta &= C^{\otimes d } \zeta,
\end{aligned}
\end{equation}
where $\zeta=x^{\otimes d} \in \mathbb{R}^{n^d}$, is the lifted state vector, $w= ux^{\otimes(d-1)} \in \mathbb{R}^{n^{(d-1)}} $ is the lifted input vector and   $ \eta=y^{\otimes d} \in \mathbb{R}$.  $A^{\oplus d} \in \mathbb{R}^{n^d \times n^d}$, $B^{\oplus d} \in \mathbb{R}^{n^d \times n^{d-1}}$ and $C^{\otimes d} \in \mathbb{R}^{1 \times n^d}$ are the lifted matrices characterizing the lifted system dynamics. 

\section{star norm and inescapable ellipsoids ($d=1$)}\label{star norm section}
The authors  in \cite{feron1994linear} and \cite{abedor1996linear} avoided the complexity on computing the $L_{\infty}$-induced norm of the LTI systems by computing the star norm, an upper bound on the $L_{\infty}$-induced norm ($\ell_1$ norm). The star norm is obtained by approximating the reachable sets with unit peak input with inescapable ellipsoids which are defined as  follows. 

\begin{definition}(Reachable set with unit peak input) \cite{feron1994linear}
is defined as the set of all reachable states from the origin in finite time by unit peak input $\|u\|_{\infty}\leq1$
\begin{equation}\label{reachable set}
   \mathcal{R}_{\mathrm{up}} \triangleq\left\{x(T) : \begin{array}{lr}
\dot{x}=Ax+Bu, \quad x(0)=0 \\
u^{T} u \leq 1, \quad T \geq 0
\end{array}\right\}
\end{equation}
\end{definition}
\begin{definition}(Inescapable set) \cite{abedor1996linear} A set $\mathcal{X}$ is said to be inescapable if (1) the origin ($x(0)\in \mathcal{X}$) and (2) for $x(0) \in \mathcal{X}$ and $u^Tu\leq 1$, $x(t)$ will evolve inside $\mathcal{X}$ for all future time $t \ge0$    
\end{definition}

\begin{theorem} \cite{abedor1996linear}
consider the stable LTI system (\ref{LTI system}). Let $P \in \mathbb{R}^{n \times n}$ be a positive semi-definite. The closed ellipsoid $\{ x: x^T P x \leq 1\}$ is an inescapable if and only if there exist $\alpha \in R_{+}$ such that: 
\begin{equation}\label{LMI inescapable ellipsoid}
    \left[\begin{array}{cc}
A^T P+P A+\alpha P & P B \\
B^T P & -\alpha I
\end{array}\right]\leq 0
\end{equation}
\end{theorem}
By using Schur complement, and letting $Q=P^{-1}$, The LMI (\ref{LMI inescapable ellipsoid}) becomes $A Q +Q A^T + \alpha Q + BB/\alpha \leq 0 $. So, there exists a unique solution if $(A+\alpha I_n/2)$ is a stable matrix. Therefore, $\alpha \in (0, \kappa)$, where $\kappa=-2 \text{max}(\text{real} (\text{eig}(A)))$. For given $\alpha$, (\ref{LMI inescapable ellipsoid}) can be solved to get the equivalent inescapable ellipsoid. The objective is to minimize the maximum output $\|Cx\|_\infty$ inside the ellipsoid $x^T P x \leq 1$, so the procedure is as follows: 
\begin{enumerate}
    \item For $\alpha$ sweeps from zero to $\kappa$, Solve the following semi-definite program.
\begin{equation}\label{Semidefinite program}
    \begin{aligned}
& \underset{P}{\text{minimize}}
& & C P^{-1} C^T \\
& \text{subject to}
& & P>0 \\
&&& \begin{bmatrix}
    PA + A^T P + \alpha P & P B \\
    B^T P & - \alpha I
    \end{bmatrix} \leq0
\end{aligned}
    \end{equation}
    \item For each $\alpha$, compute the upper bound of the output $Cx$ inside the corresponding inescapable ellipsoid $x^T P x \leq 1$, where $P$ is the solution of (\ref{Semidefinite program}).  
    \begin{equation} \label{N_alpha}
        N_{\alpha} =\underset{x^T P x\leq 1} {\text{sup}} ||Cx|| = ||C P ^{-1} C^T || ^{\frac{1}{2}}
    \end{equation}
    \item Then, the star norm is the lowest of the upper bounds
    \begin{equation}\label{star norm}
        H^{\star} =\underset{\alpha \in (0,k)}{\text{inf}} N_{\alpha}.
    \end{equation}
\end{enumerate}
The star norm is the least conservative upper bound determined by inescapable ellipsoids. 
\section{Star norm and inescapable ellipsoids ($d=2$)}\label{star norm d=2}
In this section, we introduce the main contribution of our work. The LTI system (\ref{LTI system}) is lifted to a higher order space such that $d=2$. From (\ref{ lifted system}), the lifted system is 
\begin{equation}\label{ lifted system d=2}
\begin{aligned}
    \dot{\zeta}&= A^{\oplus 2} \zeta + B^{\oplus 2} w\\
    \eta &= C^{\otimes 2 } \zeta,
\end{aligned}
\end{equation}
such that $\zeta= x^{\otimes 2} \in \mathbb{R}^{n^2}$, $w=ux \in \mathbb{R}^{n}$ and $\eta = y^{\otimes d} \in \mathbb{R}$. The matrices $A \in \mathbb{R}^{n^2 \times n^2}$, $B \in \mathbb{R}^{n^2 \times n}$ and $C \in \mathbb{R}^{1 \times n^2}$ can be written as 
\begin{equation}\label{lifted matrices d =2}
\begin{aligned}
  A^{\oplus 2} &= A \otimes I_n + I_n \otimes A \\ 
  B^{\oplus 2} &= B \otimes I_n + I_n \otimes B \\
  C^{\otimes 2}& = C \otimes C.
\end{aligned}
\end{equation}

\begin{lemma}\label{lemma}
If $A$ is a stable matrix, then $A^{\oplus 2}$ is also a stable matrix.
\end{lemma}
\begin{proof}
By using Theorem 4.4.5 in \cite{horn1994topics}. Let the eignvalues of $A$ are $\lambda_{1}, \lambda_{2}, \dots, \lambda_{n}$ and the corresponding eignvectors are $v_1,v_2,\dots,v_n$, then the eignvalues of $A^{\oplus 2} $ are the sum of each pair of the eignvalues of $A$. So, $\sigma(A)=\{\lambda_i + \lambda_j:  i =1,2,\dots,n  \text{ and } j = 1,2,\dots n\}$ and the corresponding eignvector of each sum is $v_i \otimes v_j$.
since $A$ is a stable matrix, so the real part of each $\lambda$ is negative. Therefore the real part of the sum of any pairs is also negative. So, the real parts of all eignvalues of $A^{\oplus 2}$ are negative which implies the stability of $A^{\oplus 2}$. 
\end{proof}

To clarify the idea of this section which is the main contribution of this work, we consider a second-order LTI system. Then, the idea can be generalized to higher order systems. Let $x \in \mathbb{R}^2$, then the lifted dynamics (\ref{ lifted system d=2}) will be

\begin{equation}\label{lifted system n=2 d=2}
    \begin{aligned}
 \dot{\zeta}&= \begin{bmatrix}
2a_{11} &a_{12}& a_{12} & 0 \\ 
    a_{21} & a_{11}&a_{22}& a_{12} \\
     a_{21}& a_{11} &a_{22} &a_{12} \\
     0 & a_{21} & a_{21} & 2a_{22} 
\end{bmatrix}
\zeta+ \begin{bmatrix}
2 b_1 & 0 \\
    b_1  &b_2 \\
    b_1  &b_2 \\
     0 & 2b_2 
\end{bmatrix}
w \\
\eta &= \begin{bmatrix}
c_1^2 &c_1 c_2 & c_1 c_2 & c_2^2 
\end{bmatrix}\zeta,
\end{aligned}
\end{equation}
where, $\zeta= \begin{bmatrix} x_1^2 &x_1 x_2 & x_2 x_1 & x_2^2\end{bmatrix}^T$, $w=\begin{bmatrix} ux_1 & u x_2 \end{bmatrix}^T$, and $\eta=y^2$. The objective to get the inescapable ellipsoids which approximates the reachable set with unit peak input $u^2 \leq 1$ and the star norm that is considered as an upper bound on the $L_{\infty}$- induced norm. 

The condition $u^2 \leq 1$ imposes new inequality constraints on the lifted states $\zeta$ and lifted inputs $w$ as follows 
\begin{equation}\label{inequality constraints n=2}
    \begin{aligned}
    w_1^2&= u^2 x_1^2 \leq x_1^2 = \zeta_1\\
    w_2^2&= u^2 x_2^2 \leq x_2^2 = \zeta_4.
    \end{aligned}
\end{equation}
There also exists an equality constraint, $w_1 x_2 - w_2 x_1 = 0$, but this equality is a function of $w$ and original state $x$. So, we multiply both sides by $x_1$ and $x_2$ to get two different quadratic equality constraints in $\zeta$ and $w$ given by
\begin{equation}\label{equality w and zeta}
\begin{aligned}
     w_1 \zeta_2 -w_2 \zeta_1 &= 0 \\
     w_1 \zeta_4- w_2 \zeta_3 &=0.
\end{aligned}
\end{equation}

Suppose there exist a quadratic function $V(\zeta) = \zeta^T P \zeta $ where $P\in \mathbb{S}_{++}^{4 \time 4} \succ 0$ and $dV(\zeta)/dt \leq 0$ for all $\zeta$ and $w$ satisfying (\ref{lifted system n=2 d=2}) whenever $\zeta^T P \zeta \ge 1$ and $w, \zeta$ satisfy (\ref{inequality constraints n=2}) and (\ref{equality w and zeta}). Then, the ellipsoid $\zeta^T P \zeta \leq 1$ is an inescapable ellipsoid and contains the reachable set. Therefore, these conditions can be written as
\begin{equation}\label{conditions n=2 d=2}
  \begin{aligned}
  &\zeta^T (PA^{\oplus2}+{A^{\oplus 2}}^T P ) \zeta + \zeta^T P B^{\oplus 2} w + w^T {B^{\oplus 2}}^T P \zeta \leq 0\\
  &\forall \zeta \text{ and } w \text{ whenever, } \zeta^T P \zeta \leq 1, \text{ and }\\
  &\qquad \qquad \qquad \qquad \qquad (\ref{inequality constraints n=2}) \text{ and } (\ref{equality w and zeta}) \text{ are satisfied}.
  \end{aligned}  
\end{equation}
Using $\mathcal{S}$ procedure \cite{feron1994linear}, condition (\ref{conditions n=2 d=2}) holds if there exists $\alpha \ge 0$, $\beta_1 \ge 0$, $\beta_2 \ge 0$, $\gamma_1$ and $\gamma_2$ such that for all $\zeta$ and $w$:
\small{
\begin{equation}\label{LMI n=2 d=2}
    \begin{bmatrix}
    PA^{\oplus2}+{A^{\oplus 2}}^T P + \alpha P & P B^{\oplus 2} + \gamma_1 E_1 + \gamma_2 E_1 & z \\
    {B^{\oplus 2}}^T P+\gamma_1 E_1^T + \gamma_2 E_2^T & W &\mathbf{0}_2 \\
    z^T  &\mathbf{0}_2^T&-\alpha
    \end{bmatrix}\le 0
\end{equation}} \normalsize

\noindent where $\mathbf{0}_2 \in \mathbb{R}^2$ is the zero vector in $\mathbb{R}^2$,  $z= \begin{bmatrix} \beta_1/2 &0&0& \beta_2/2 \end{bmatrix}^T$, and the matrices $E_1$, $E_2$ and W are 
\begin{equation}\label{E_1,E_2 and W}
\begin{aligned}
    E_1= \begin{bmatrix}
    0 &-1 \\
    1 &0 \\
    0 &0 \\
    0 &0
    \end{bmatrix}, \quad 
    E_1= \begin{bmatrix}
    0 &0 \\
    0 &0 \\
    0 &-1 \\
    1 &0
    \end{bmatrix}, \quad
    W= \begin{bmatrix}
    -\beta_1 & 0 \\0 &-\beta_2
    \end{bmatrix}.
\end{aligned}
\end{equation}
The same procedure used in section \ref{star norm section} to get the inescapable ellipsoid and the star norm for the lifted system. First, the following SDP is solved for every $\alpha \in (0,\kappa)$ , where $\kappa = -2 \text{max}(\text{real} (\text{eig}(A^{\oplus 2})))$
\begin{equation}\label{Semidefinite program n=2}
    \begin{aligned}
& \underset{P}{\text{minimize}}
& & C^{\otimes 2} P^{-1} {C^{\otimes 2}}^T \\
& \text{subject to}
& & P>0, \quad (\ref{LMI n=2 d=2}).
\end{aligned}
    \end{equation}
Then, find the upper bound of the output in every inescapable ellipsoid $x^{\otimes 2} P x^{\otimes 2} \leq 1$ as
\begin{equation} \label{N_alpha n=2 d=2}
        N_{\alpha} = ||C^{\otimes 2} P ^{-1} {C^{\otimes 2}} || ^{\frac{1}{2}}.
    \end{equation}
The star norm for the lifted system is the smallest of these upper bounds. 
\begin{equation}\label{star norm n=2 d=2}
        H_2^{\star} =\underset{\alpha \in (0,k)}{\text{inf}} N_{\alpha}
    \end{equation}
Finally, the star norm of the original system can be calculated as the square root of the star norm obtained for the higher order system ($d=2$). 

The lifting procedure can be done for higher order systems ($n>2$) using the same procedure, but the number of equality and inequality constraints imposed in the $\mathcal{S}$ procedure will increase. In this case, we have $n$ inequality constraints $\{w_i^2 \leq \zeta_{i(n+1)-n}: i=1\dots n\}$. The number of equality constraints is $n(n-1)/2$:
\begin{equation}
\begin{matrix}
    w_1x_2-w_2x_1=0, \dots, w_1x_n-w_n x_1=0,\\
    w_2x_3-w_3x_2=0, \dots, w_2x_n-w_nx_2=0,\\
    \dots w_{n-1}x_n-w_nx_{n-1}=0
\end{matrix}
\end{equation}
As in (\ref{equality w and zeta}) each of both sides of these equality constraints is multiplied by $x_1,x_2,\dots, x_n$ to come up with $n^2(n-1)/2$ quadratic equality constraints in $w$ and $\zeta$ that can be imposed easily in the $\mathcal{S}$ procedure. 

\section{Numerical experiments}
In this section, we consider three systems: high damping, low damping and stiff systems. To show the effect of lifting the dynamical system and using homogeneous Lyapunov function in approximating the reachable sets and computing upper bounds on the $\ell_1$ norm. 
\subsection{System with high damping}
Consider the following LTI system 
\begin{equation}\label{high damping system}
    \dot{x}=\begin{bmatrix}
    0& 1\\-4 &-4
    \end{bmatrix}x+\begin{bmatrix}
    0\\1
    \end{bmatrix}u, \qquad y = \begin{bmatrix}
    1&1
    \end{bmatrix}x
\end{equation}
The exact value of the $\ell_1$ norm is obtained using Mathematica by computing $\int_0^{\infty} |C e^{At} B| dt=\mathbf{0.3177}$. By using method introduced in section \ref{star norm section} without lifting ($d=1$), the computed star norm is $\mathbf{0.3536}$. After lifting the system to $d=2$, the star norm computed using the method in section \ref{star norm d=2} is $\mathbf{0.3368}$ which is a better approximation for the upper bound of the $\ell_1$ norm. Additionally, in Fig.  \ref{fig:high damping system sets}, the blue and black ellipsoids represents the approximate of the unit peak input reachable sets using ($d=1$) and ($d=2$) respectively. Which shows that using the higher order system produces better approximation for the reachable set. 
The red line in Fig. \ref{fig:high damping system sets} represents an attempt at computing a worst case trajectory, the trajectory generated using the control input $u$ that maximize the gradient of the Lyapunov function ($dV/dt$) at every time step. this control input can be written as 
\begin{equation}\label{worst case control}
\begin{aligned}
    u &= \underset{u^2 \leq 1}{\text{argmax}} \quad x^T(PA+A^T P)x+ 2 x^T P B u \\
    &= \text{sign} (x^T P B)
\end{aligned}
\end{equation}
\begin{figure}[thpb]
    \centering
    \includegraphics[scale=0.5]{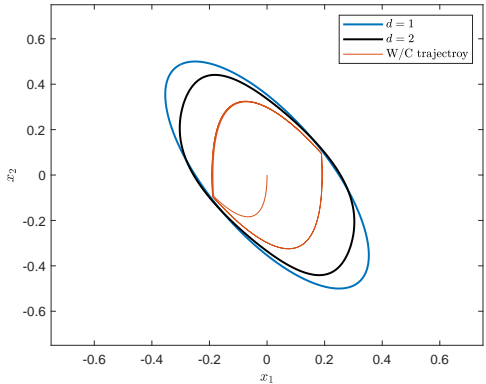}
    \caption{Approximate unit peak peak input reachable sets and worst case trajectory for the high damping system (\ref{high damping system})  . }
    \label{fig:high damping system sets}
\end{figure}

The peak output of the worst case trajectory represents the lower bound for the $\ell_1$ norm  Fig. \ref{fig:high damping system trajectory} shows that value of this lower bound is $\mathbf{0.3097}$
\begin{figure}[thpb]
    \centering
    \includegraphics[scale=0.5]{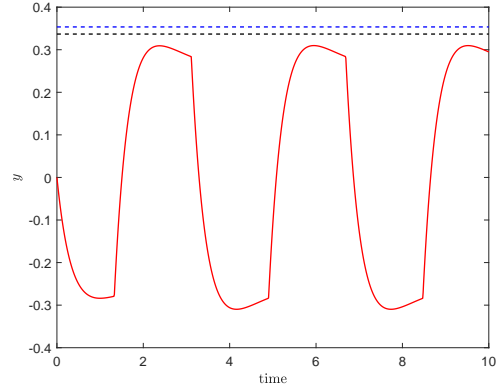}
    \caption{Red curve represents the "Worst case" trajectory and the blue and  black dotted lines are the star norm for $d=1$ and $d=2$ respectively for  the high damping system (\ref{high damping system}) }
    \label{fig:high damping system trajectory}
\end{figure}
\subsection{System with low damping}
 Consider the following low damping system
\begin{equation}\label{low damping system}
    \dot{x}=\begin{bmatrix}
    0& 1\\-0.5 &-0.5
    \end{bmatrix}x+\begin{bmatrix}
    0\\1
    \end{bmatrix}u, \qquad y = \begin{bmatrix}
    1&1
    \end{bmatrix}x
\end{equation}
The exact value of the $\ell_1$ norm is by calculating the integral $\int_0^{\infty} |C e^{At} B| dt$ is $\mathbf{4.3069}$.
The computations shows that the value of the star norm at ($d=2$), $\mathbf{4.5533}$, is less than the value of star norm obtained at ($d=1$), $\mathbf{4.63}$. In addition, Fig. \ref{fig:low damping system sets} shows that the approximation of the unit peak input reachable set is less conservative when using the lifted system ($d=2$). 
\begin{figure}[thpb]
    \centering
    \includegraphics[scale=0.5]{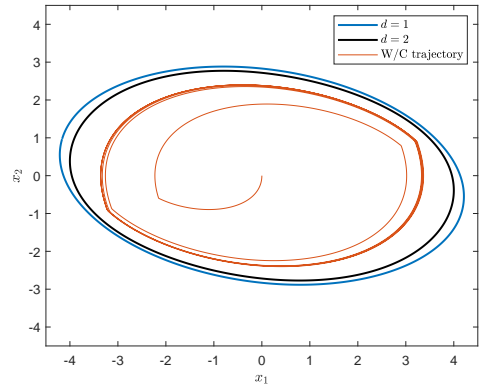}
    \caption{Approximate unit peak peak input reachable sets and worst case trajectory for low damping system (\ref{low damping system}) . }
    \label{fig:low damping system sets}
\end{figure}

\subsection{stiff system}
Consider the following system 
\begin{equation}\label{stiff system}
    \dot{x}=\begin{bmatrix}
    -1& 0\\0 &-100
    \end{bmatrix}x+\begin{bmatrix}
    1\\100
    \end{bmatrix}u, \qquad y = \begin{bmatrix}
    1&-2
    \end{bmatrix}x
\end{equation}
The exact value of the $\ell_1$ norm is $\mathbf{3.0412}$. The star norm computed when ($d=1$) is $\mathbf{ 10.4600}$ which is very conservative upper bound for the $\ell_1$ norm. However, the star norm computed using the lifted system is $\mathbf{5.7680}$ which shows the effectiveness of lifting the system to obtain better bounds. Also, Fig. \ref{fig:stiff system sets} shows the significant effect of lifting in reducing the conservatism in reachable set approximations for system (\ref{stiff system}). 
\begin{figure}[thpb]
    \centering
    \includegraphics[scale=0.5]{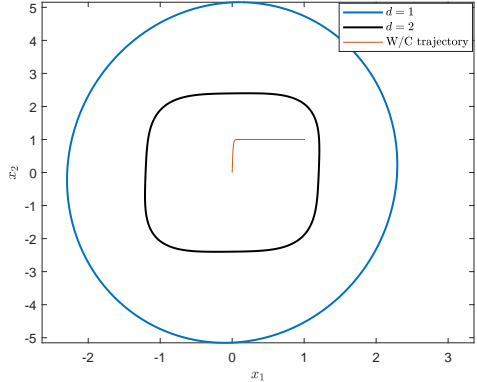}
    \caption{Approximate unit peak peak input reachable sets and worst case trajectory for stiff system (\ref{stiff system})  . }
    \label{fig:stiff system sets}
\end{figure}

\section{Alternative method to approximate $\ell_1$ norm}
In this section, we introduce another method to get better and more accurate approximation for the $\ell_1$ norm which is $\int_0^{\infty} |C e^{At} B |dt$. This integral is divided in two parts as:
\begin{equation}\label{integral split}
    \int_0^{\infty} |C e^{At} B |dt=\int_0^{T_0} |C e^{At} B |dt + \int_{T_{0}}^{\infty}\left|C e^{A t} B\right| d t
\end{equation}
where $T_0 >0$ is known, so the first term can be computed exactly. The second term is the $\ell_1$ norm of the following system: 
\begin{equation}\label{equivalent system}
    \begin{aligned}
\dot{z} &=A z+e^{A T_{0}} B u \\
y &=C z
\end{aligned}
\end{equation}
An upper bound on the $\ell_1$ norm of system (\ref{equivalent system}) can be approximated as the star norm where the system can be lifted to a higher order to obtain better approximation. 
Table \ref{table1} shows the upper bound of the $\ell_1$ norm of the low damping system (\ref{low damping system}) using different values of $T_0$. At each $T_0$, the first term in (\ref{integral split}) is computed and the second term is approximated at $d=1$ and $d=2$ for the system (\ref{equivalent system}). 

\begin{table}[h]
\caption{$\ell_1$ norm approximation for system (\ref{low damping system}) [exact value is $4.3069$]} 
\label{table1}
\begin{center}
\begin{tabular}{|c|c|c|}
\hline
$T_0 \text(Sec)$ & Upper bound ($d=1$)& Upper bound ($d=2$)\\
\hline
2 & 4.5683 &4.5304\\
\hline
5 & 4.4078 &4.3981\\
\hline
10 & 4.3376 &4.3332\\
\hline
20 & 4.3096 &4.3091\\
\hline
\end{tabular}
\end{center}
\end{table}
This method provides more accurate approximations even for stiff system (\ref{stiff system}), The upper bound obtained by using $T_0=0.05 \text{sec}$ and $d=2$ is $\mathbf{3.0424}$ which is very close to the exact value $\mathbf{3.0412}$
\section{CONCLUSIONS}
This work demonstrates how the  polynomial homogeneous Lyapunov function can be used to obtain better approximations for the LTI system's star norm which is an upper bound for the system's $\ell_1$ norm. We also demonstrated that the inescapable ellipsoids, approximations for the unit peak input reachable sets,  generated using higher order Lyapunov functions is more accurate and less conservative that those produced by using quadratic Lyapunov functions. We tested that for three types of systems: systems with high damping, systems with low damping and stiff systems, to show the improvement of the reachable sets approximation. We only considered lifting the LTI system $d=2$. For $d>2$, the main difficulty is to transfer the inequality $u^2 \leq 1$ in the original space to quadratic inequalities in the higher order space as in (\ref{inequality constraints n=2}) for $d=2$. As a future work, we are working to develop new techniques to deal with this difficulty. Also, there are future opportunities to use homogeneous Lyapunov functions in robustness analysis and design better-performing controllers.
\addtolength{\textheight}{-12cm}   

\bibliographystyle{IEEEtran}
\bibliography{ECC2022.bib}

\begin{thebibliography}{10}
\providecommand{\url}[1]{#1}
\csname url@samestyle\endcsname
\providecommand{\newblock}{\relax}
\providecommand{\bibinfo}[2]{#2}
\providecommand{\BIBentrySTDinterwordspacing}{\spaceskip=0pt\relax}
\providecommand{\BIBentryALTinterwordstretchfactor}{4}
\providecommand{\BIBentryALTinterwordspacing}{\spaceskip=\fontdimen2\font plus
\BIBentryALTinterwordstretchfactor\fontdimen3\font minus
  \fontdimen4\font\relax}
\providecommand{\BIBforeignlanguage}[2]{{%
\expandafter\ifx\csname l@#1\endcsname\relax
\typeout{** WARNING: IEEEtran.bst: No hyphenation pattern has been}%
\typeout{** loaded for the language `#1'. Using the pattern for}%
\typeout{** the default language instead.}%
\else
\language=\csname l@#1\endcsname
\fi
#2}}
\providecommand{\BIBdecl}{\relax}
\BIBdecl

\bibitem{shamma1993nonlinear}
J.~S. Shamma, ``Nonlinear state feedback for {$\ell_1$} optimal control,''
  \emph{Systems \& control letters}, vol.~21, no.~4, pp. 265--270, 1993.

\bibitem{abedor1996linear}
J.~Abedor, K.~Nagpal, and K.~Poolla, ``A linear matrix inequality approach to
  peak-to-peak gain minimization,'' \emph{International Journal of Robust and
  Nonlinear Control}, vol.~6, no. 9-10, pp. 899--927, 1996.

\bibitem{feron1994linear}
E.~M. Feron, ``Linear matrix inequalities for the problem of absolute stability
  of control systems,'' Ph.D. dissertation, Stanford University, 1994.

\bibitem{vidyasagar1986optimal}
M.~Vidyasagar, ``Optimal rejection of persistent bounded disturbances,''
  \emph{IEEE Transactions on Automatic Control}, vol.~31, no.~6, pp. 527--534,
  1986.

\bibitem{boyd1987comparison}
S.~Boyd and J.~Doyle, ``Comparison of peak and $\text{RMS}$ gains for
  discrete-time systems,'' \emph{Systems \& control letters}, vol.~9, no.~1,
  pp. 1--6, 1987.

\bibitem{balakrishnan1992computing}
V.~Balakrishnan and S.~Boyd, ``On computing the worst-case peak gain of linear
  systems,'' \emph{Systems \& Control Letters}, vol.~19, no.~4, pp. 265--269,
  1992.

\bibitem{dahleh19871}
M.~Dahleh and J.~Pearson, ``{$\ell_1$} optimal feedback controllers for
  {$MIMO$} discrete time systems,'' \emph{IEEE Transactions on Automatic
  Control}, vol.~32, no.~4, pp. 314--322, 1987.

\bibitem{dahleh19872}
------, ``{$\ell_1$} optimal compensators for continuous-time systems,''
  \emph{IEEE Transactions on Automatic Control}, vol.~32, no.~10, pp. 889--895,
  1987.

\bibitem{shamma1996optimization}
J.~S. Shamma, ``Optimization of the ${L_{\infty}}$-induced norm under full
  state feedback,'' \emph{IEEE Transactions on Automatic Control}, vol.~41,
  no.~4, pp. 533--544, 1996.

\bibitem{abate2020lyapunov}
M.~Abate, C.~Klett, S.~Coogan, and E.~Feron, ``{L}yapunov differential equation
  hierarchy and polynomial {L}yapunov functions for switched linear systems,''
  in \emph{2020 American Control Conference (ACC)}.\hskip 1em plus 0.5em minus
  0.4em\relax IEEE, 2020, pp. 5322--5327.

\bibitem{abate2021pointwise}
------, ``Pointwise-in-time analysis and non-quadratic {L}yapunov functions for
  linear time-varying systems,'' in \emph{2021 American Control Conference
  (ACC)}.\hskip 1em plus 0.5em minus 0.4em\relax IEEE, 2021, pp. 3550--3555.

\bibitem{broxson2006kronecker}
B.~J. Broxson, ``The kronecker product,'' 2006.

\bibitem{van2000ubiquitous}
C.~F. Van~Loan, ``The ubiquitous kronecker product,'' \emph{Journal of
  computational and applied mathematics}, vol. 123, no. 1-2, pp. 85--100, 2000.

\bibitem{horn1994topics}
R.~A. Horn, R.~A. Horn, and C.~R. Johnson, \emph{Topics in matrix
  analysis}.\hskip 1em plus 0.5em minus 0.4em\relax Cambridge university press,
  1994.

\end{thebibliography}

\end{document}